\documentclass[copyright,creativecommons]{eptcs}

\usepackage{underscore}           
\usepackage{pdfsync}
\usepackage{hyperref}
\usepackage{microtype}

\usepackage[utf8]{inputenc}
\usepackage{CJKutf8}

\usepackage{amsthm}
\usepackage{amsmath}
\usepackage{thm-restate}
\usepackage{amssymb}
\usepackage[textsize=small]{todonotes}
\usepackage{enumerate}
\usepackage{xspace}
\usepackage{paralist}
\usepackage{wrapfig}
\usepackage{extpfeil}
\usepackage[capitalize]{cleveref}

\newcommand{\ignore}[1]{}

\newcommand{\goesto}[1]{\xrightarrow{#1}}
\newcommand{\lang}[1]{L(#1)}
\newcommand{\from}{\leftarrow}

\newcommand{\A}{\mathcal A}
\newcommand{\B}{\mathcal B}

\newcommand{\N}{\mathbb N}
\newcommand{\Z}{\mathbb Z}
\newcommand{\Q}{\mathbb Q}

\newcommand{\R}{\mathbb R}

\newcommand{\sep}{\mid}
\newcommand{\tuple}[1]{(#1)}

\newcommand{\set}[1]{\left\{ #1 \right\}}

\newcommand\limplies\to

\newif\ifstartedinmathmode
\newcommand*{\st}{
  \relax\ifmmode\startedinmathmodetrue\else\startedinmathmodefalse\fi
  \ifstartedinmathmode{\;.\;}\else{s.t.~}\fi%
}

\newcommand{\ie}{i.e.\xspace}

\newcommand{\card}[1]{|#1|}

\newcommand{\NLOGSPACE}{\textsf{NLOGSPACE}\xspace}
\newcommand{\PTIME}{\textsf{PTIME}\xspace}
\newcommand{\NP}{\textsf{NP}\xspace}
\newcommand{\PSPACE}{\textsf{PSPACE}\xspace}
\newcommand{\PSPACEc}{\textsf{PSPACE}-c.\xspace}
\newcommand{\EXPTIME}{\textsf{EXPTIME}\xspace}
\newcommand{\EXPTIMEc}{\textsf{EXPTIME}-c.\xspace}
\newcommand{\UNDEC}{undec.}
\newcommand{\NCk}[1]{$\textsf{NC}^{#1}$}
\newcommand{\UUCFL}{\textsf{UUCFG}\xspace}
\newcommand{\UUCFG}{\textsf{UUCFG}\xspace}
\newcommand{\SQRTSUM}{\textsf{SQRTSUM}\xspace}

\newcommand{\DFA}{\textrm{DFA}\xspace}
\newcommand{\NFA}{\textrm{NFA}\xspace}
\newcommand{\UFA}{\textrm{UFA}\xspace}
\newcommand{\DCFG}{\textrm{DCFG}\xspace}
\newcommand{\UCFG}{\textrm{UCFG}\xspace}
\newcommand{\CFG}{\textrm{CFG}\xspace}
\newcommand{\CFGs}{\textrm{CFGs}\xspace}

\newcommand{\DCFL}{\textrm{DCFL}\xspace}

\newcommand{\SCFG}{\textrm{SCFG}\xspace}

\newcommand{\SH}{\PTIME \cite{StearnsHunt:Unambiguous}}

\newcommand{\abs}[1]{\left|#1\right|}

\newcommand{\divides}{\,|\,}
\newcommand{\op}{\mathsf{op}}
\newcommand{\conv}{*}
\newcommand\convpoly[2]{#1_{\conv}[#2]} 
\newcommand\poly[2]{#1[#2]} 
\newcommand\coefficient[2]{[#1]#2}

\newcommand{\stkout}[1]{\ifmmode\text{\sout{\ensuremath{#1}}}\else\sout{#1}\fi}

\newtheorem{theorem}{Theorem}
\newtheorem{lemma}[theorem]{Lemma}

\theoremstyle{remark}
\newtheorem{claim}[theorem]{Claim}
\newtheorem*{claim*}{Claim}

\title{On the complexity of the universality and inclusion problems for unambiguous context-free grammars \\ (Invited Paper)}
\author{Lorenzo Clemente
\institute{Department of Mathematics, Informatics, and Mechanics (MIMUW)}
\institute{University of Warsaw (Warsaw, Poland)%
\thanks{This work has been partially supported by the Polish NCN grant 2017/26/D/ST6/00201. This is a technical report of an invited paper of the same title to appear in the proceedings of VPT 2020.}}
\email{clementelorenzo@gmail.com}
}
\def\titlerunning{Complexity of the universality problem for unambiguous context-free grammars}
\def\authorrunning{L.~Clemente}

\begin{document}
\maketitle

\begin{abstract}
	We study the computational complexity of universality and inclusion problems
	for unambiguous finite automata and context-free grammars.
	We observe that several such problems can be reduced to the universality problem for unambiguous context-free grammars.
	The latter problem has long been known to be decidable
	and we propose a \PSPACE algorithm
	that works by reduction to the zeroness problem of recurrence equations with convolution.
	We are not aware of any non-trivial complexity lower bounds.
	However, we show that computing the coin-flip measure of an unambiguous context-free language,
	a quantitative generalisation of universality,
	is hard for the long-standing open problem \SQRTSUM.
\end{abstract}

\section{Introduction}

The purpose of this note is to attract attention
to a long-standing open problem in formal language theory.
The problem in question is the exact complexity of deciding universality of unambiguous context-free grammars (\UUCFG).
A context-free grammar is \emph{unambiguous} if every accepted word
admits a unique parse tree,
and the universality problems asks,
for a given grammar $G$ over a finite set of terminals $\Sigma$ (alphabet),
whether $G$ accepts every word $\lang G = \Sigma^*$.
While the universality problem for context-free grammars is undecidable \cite{HopcroftUllman:1979},
the same problem for unambiguous grammars is long-known to be decidable
(a corollary of \cite[Theorem 5.5]{SalomaaSoittola:Book:PowerSeries:1978}),
e.g., by reducing to the first-order theory of the reals with one quantifier alternation
\cite[eq.~(3), page 149]{SalomaaSoittola:Book:PowerSeries:1978}.
Since the latter fragment is decidable in \EXPTIME \cite{Grigorev:JSC:1988},
this yields an \EXPTIME upper bound for \UUCFG.
%
%
%
%
No non-trivial lower bound for \UUCFG seems to be known in the literature.
%

The typical way to solve a containment problem of the form $L \subseteq M$
is to complement $M$ and solve $L \cap (\Sigma^* \setminus M) = \emptyset$.
For instance, when $L$ is regular and $M$ is deterministic context-free (\DCFG),
this gives a \PTIME procedure since \DCFG languages are efficiently closed under complement
and intersection with regular languages,
and their emptiness problem is in \PTIME.
However, \UCFG languages are not closed under complement
(the complement is not even context-free in general \cite{Hibbard:Ullian:UCFL:1966}),
so the language-theoretic approach is not available.
As Salomaa and Soittola remark in their book from 1978,
``no proof is known for Theorem 5.5 which uses only standard formal language theory''.
To this day, we are not aware of a proof of decidability for \UUCFG using different techniques%
\footnote{
    In a later book, Kuich and Salomaa reprove decidability \cite[Corollary 16.25]{KuichSalomaa:1986}
    by using variable elimination,
    which is arguably closer to algebraic geometry than formal languages.
}.
The \UUCFG problem is not isolated in this respect.

\paragraph{State of the art.}

Let $\A, \B$ be two classes of language acceptors.
Examples include deterministic (\DFA), unambiguous (\UFA), and nondeterministic finite automata (\NFA),
and similarly for context-free grammars we have the classes \DCFG, \UCFG, and \CFG.
The ``$\A \subseteq \B$'' \emph{inclusion problem} asks, 
given a language acceptor $A$ from $\A$ and $B$ from $\B$,
whether the languages they recognise satisfy $\lang A \subseteq \lang B$.
A summary of decidability and complexity result for inclusion problems
involving finite automata and grammars
is presented in \cref{fig:table}.
Many entries in the table are well-known.
The problem $\NFA \subseteq \NFA$ is a classic \PSPACE-complete problem \cite{MeyerStockmeyer:Equivalence:1972}.
The problem $\UFA \subseteq \UFA$ was shown in \PTIME
by Stearns and Hunt in their seminal paper \cite{StearnsHunt:Unambiguous}%
\footnote{
    An incomparable \NCk 2 upper bound for this problem is also known
    \cite[Fact 4.5]{MassazzaSabadini:TAPSOFT:1989} (c.f.~\cite[Theorem 2]{Tzeng:IPL:1996}).
}.
The fact that $\CFG \subseteq \NFA$ is \EXPTIME-complete is somewhat less known \cite[Theorem 2.1]{Kasai:Iwata:1992}.
%
%
%
%
%
The inclusion problems $\A \subseteq \UFA$ when $\B$ is $\DCFG$, $\UCFG$, or $\CFG$
do not appear to have been studied before.
The $\A \subseteq \B$ problem is undecidable as soon as both $\A, \B$ are context-free grammars,
since $\DCFG \subseteq \DCFG$ is well-known to be undecidable \cite[Theorem 10.7, Point 2]{HopcroftUllman:1979}.
We have already observed that $\NFA \subseteq \DCFG$ is in \PTIME.
The equivalence problem $\NFA = \UCFG$
is shown to be decidable in \cite[Theorem 5.5]{SalomaaSoittola:Book:PowerSeries:1978},
although no complexity bound is given.
The more general inclusion $\NFA \subseteq \UCFG$ does not seem to have been studied before.

%

%


\begin{figure}
    \begin{center}
    	\begin{tabular}{l|l|l|l|l|l|l}
    		${\subseteq}$\hfill	& \DFA		& \UFA			& \NFA 		& \DCFG		& \UCFG		& \CFG \\
    		\hline
    		\DFA				& \PTIME	& \PTIME		& \PSPACEc \cite{MeyerStockmeyer:Equivalence:1972}	& \PTIME	& =\UUCFL (Th.~\ref{thm:NFA:UCFG})	& \UNDEC \\
    		\UFA				& \PTIME    & \SH 			& \PSPACEc \cite{MeyerStockmeyer:Equivalence:1972}	& \PTIME	& =\UUCFL (Th.~\ref{thm:NFA:UCFG})	& \UNDEC \\
    		\NFA				& \PTIME	& \PTIME (Th.~\ref{thm:NFA:UFA}) & \PSPACEc \cite{MeyerStockmeyer:Equivalence:1972}	& \PTIME	& =\UUCFL  (Th.~\ref{thm:NFA:UCFG})	& \UNDEC \\
    		\DCFG				& \PTIME	& $\leq$\UUCFL	(Th.~\ref{thm:CFG:UFA}) & \EXPTIMEc	\cite{Kasai:Iwata:1992} & \UNDEC	& \UNDEC	& \UNDEC \\
    		\UCFG				& \PTIME	& $\leq$\UUCFL	(Th.~\ref{thm:CFG:UFA})& \EXPTIMEc \cite{Kasai:Iwata:1992} & \UNDEC    & \UNDEC	& \UNDEC \\
    		\CFG				& \PTIME	& $\leq$\UUCFL (Th.~\ref{thm:CFG:UFA})	& \EXPTIMEc	\cite{Kasai:Iwata:1992} & \UNDEC	& \UNDEC	& \UNDEC
    	\end{tabular}
	\end{center}
	
	\noindent
	``$\leq$\UUCFL'': the problem reduces in \PTIME to \UUCFL.
	
	\noindent
	``$=$\UUCFL'': the problem is \PTIME inter-reducible with \UUCFL.
	
	\caption{Inclusion problems for various classes of regular and context-free languages.}
	\label{fig:table}
\end{figure}

\paragraph{Contributions.}

We establish several connections between inclusion problems $\A \subseteq \B$
when $\B$ is \UFA or \UCFG with the \UUCFG problem.
Our contributions are as follows.

\begin{enumerate}[1.]

    \item We observe that in many cases the inclusion problem $L \subseteq M$
    reduces in polynomial time to the sub-case where $L$ is deterministic (\cref{sec:det:trick}).
    One application is lower bounds:
    Once we know that $\CFG \subseteq \NFA$ is \EXPTIME-hard \cite[Theorem 2.1]{Kasai:Iwata:1992},
    we can immediately deduce that the same lower bound carries over to $\DCFG \subseteq \NFA$ \cite[Theorem 3.1]{Kasai:Iwata:1992}.
    
    \item We observe that in many cases the inclusion problem $L \subseteq M$ with $L$ deterministic
    reduces in polynomial time to the universality problem (\cref{sec:incl2univ}).
    One application is upper bounds (combined with the previous point):
    For instance, from the fact that  $\UFA = \Sigma^*$ is in \PTIME
    we can deduce that the more general problem $\NFA \subseteq \UFA$ is also in \PTIME (\cref{thm:NFA:UFA}),
    which seems to be a new observation.
    
    \item We apply the last two points to show that the following inclusion problems $\A \subseteq \B$ reduce to \UUCFG:
    $\A \in \set{\DCFG, \UCFG, \CFG}$ and $\B = \UFA$ (\cref{thm:CFG:UFA});
    $\A \in \set{\DFA, \UFA, \NFA}$ and $\B = \UCFG$ (\cref{thm:NFA:UCFG}).
    Since \UUCFG is a special instance of the latter set of problems,
    they are \PTIME inter-reducible with \UUCFG.
    
    \item We show that \UUCFG is in \PSPACE (\cref{thm:UUCFL:PSPACE}),
    which improves the \EXPTIME upper bound that can be extracted from \cite{SalomaaSoittola:Book:PowerSeries:1978}.
    A \PSPACE upper bound for the same problem
    has also been shown by S.~Purgał in his master thesis \cite[Section 3.7]{Purgal:MSC:2018}.
    
    \item We complement the upper bound in the previous point
    by showing that computing the so-called \emph{coin-flip measure} of a \UCFG
    (a quantitative problem generalising universality; c.f.~\cref{sec:coin-flip})
    is \SQRTSUM-hard (\cref{thm:SQRTSUM-hardness}).
    The latter is a well-known problem in the theory of numerical computation,
    which is not known to be in \NP or \NP-hard
    \cite{AllenderBurgisserKjeldgaard-PedersenMiltersen:JC:2009,EtessamiYannakakis:JACM:2009}.
    
\end{enumerate}
The generic and simple polynomial time reductions of points 1.~and 2.~above
do not seem to be known in the literature.
Beyond the seminal work on $\UFA$ \cite{StearnsHunt:Unambiguous},
they also apply to very recent contributions on expressive models such as unambiguous register automata (c.f.~\cite{MottetQuaas:STACS:2019} for equality atoms)
and unambiguous finite and pushdown Parikh automata \cite{BostanCarayolKoechlinNicaud:ICALP:2020}.
In each of the cases above,
one can reduce from inclusion to universality.
A non-example where the reduction cannot be applied
is unambiguous Petri-nets with coverability semantics \cite{CzerwinskiFigueiraHofman:UVASS:2020}.


The \PSPACE upper bound on \UUCFG is obtained by reduction to a more general counting problem interesting on its own.
We introduce a natural class of number sequences ${f : \N \to \N}$
which we call \emph{convolution recursive} (conv-rec).
Examples include the Fibonacci $F(n+1) = F(n) + F(n-1)$
and Catalan numbers ${C(n+1) = (C \conv C)(n)}$,
where ``$\conv$'' denotes the convolution product.
We show that the function counting the number of words in $\lang G$ of a given length is conv-rec if $G$ is \UCFG.
(This result is analogous to the well-known
fact that \UCFG have algebraic generating functions
\cite{ChomskySchutzenberger:Algebraic:1963}.)
The \emph{zeroness problem} asks whether such a sequence is identically zero.
Our last contribution is a complexity upper-bound for the zeroness problem of conv-rec sequences.
\begin{enumerate}[6.]
    \item We show that the zeroness problem of conv-rec sequences is in \PSPACE (\cref{thm:zeroness}).
    We express this problem with a formula in the existential fragment for first-order logic over the reals,
    which can be decided in \PSPACE \cite{Canny:1988:STOC:1988}.
    %
    %
\end{enumerate}

\section{Convolution recursive sequences and their zeroness problem}

\paragraph{Convolution recursive sequences.}
Let $\N$, $\Z$, $\Q$, and $\R$ be the sets of natural, resp., integer, rational, and real numbers.
Let $\poly \Q {x_1, \dots, x_k}$ denote the ring of polynomials with coefficients from $\Q$
and variables $x_1, \dots, x_k$.
For two sequences indexed by natural numbers $f, g : \N \to \R$,
their \emph{sum} $f + g$ is the sequence $(f+g)(n) = f(n) + g(n)$,
and their \emph{convolution} is the sequence $(f \conv g)(n) = \sum_{k=0}^n f(k) \cdot g(n-k)$.
The convolution operation is associative $f \conv (g \conv h) = (f \conv g) \conv h$,
commutative $f \conv g = g \conv f$,
has as (left and right) identity the sequence $1,0,0,\dots$,
and distributes over the sum operation $(f + g) \conv h = f \conv g + g \conv h$.
Thus, sequences with the operations ``$+$'' and ``$\conv$'' form a semiring.
Let $\sigma : (\N \to \R) \to (\N \to \R)$ be the \emph{(forward) shift operator} on sequences,
which is defined as $(\sigma f)(n) = f(n+1)$.
The \emph{zeroness} problem for a sequence $f : \N \to \R$
amounts to decide whether $f(n) = 0$  for every $n \in \N$.

A \emph{convolution polynomial} $p(x_1, \dots, x_k)$
is a polynomial where the multiplication operation is interpreted as convolution
and a constant $k \in \Q$ is interpreted as the sequence $k,0,0,\dots$.
For example, $4 \conv (x_1 \conv x_2) + 3 \conv (x_2 \conv x_2)$
is a convolution polynomial of two variables $x_1, x_2$.
Let $\convpoly \Q {x_1, \dots, x_k}$ denote the ring of convolution polynomials with variables $x_1, \dots, x_k$.
A sequence $f : \N \to \R$ is \emph{convolution recursive} (\emph{conv-rec})
if there are $k$ \emph{auxiliary sequences} $f_1, \dots, f_k : \N \to \R$ with $f_1 = f$
and $k$ convolution polynomials $p_1, \dots, p_k \in \convpoly \Q {x_1, \dots, x_k}$
s.t., 
\begin{align}
	\label{eq:convrec}
	\left\{
		\begin{array}{lcl}
			\sigma f_1 &=& p_1(f_1, \dots, f_k), \\
			&\vdots& \\
			\sigma f_k &=& p_k(f_1, \dots, f_k).
		\end{array}
	\right.
\end{align}
The \emph{combined degree} of the representation above is the sum of the degrees of $p_1, \dots, p_k$.
For example, the Catalan numbers $C : \N \to \N$
are conv-rec (of combined degree two)
since $(\sigma C)(n) = (C \conv C)(n)$.

\begin{lemma}
	\label{lem:bounded-ratio}
	Let $f : \N \to \R$ be a conv-rec sequence of combined degree $\leq d$.
	Then ${\lim_{n \to \infty} \frac {f(n+1)} {f(n)} = O(d)}$.
\end{lemma}

\begin{proof}
	The maximal relative growth $\frac {f(n+1)} {f(n)}$ of a conv-rec sequence
	is achieved when $f$ satisfies a recurrence of the form
	%
		$\sigma f = {f \conv \cdots \conv f}$ ($d$ times)
	%
	for some degree $d \in \N$.
	If $f(0) = 1$, then the resulting sequence is known as the Fuss-Catalan numbers
	\cite{GrahamKnuthPatashnik:CM:1994}
	and it equals $f(n) = {d \cdot n + 1 \choose n} \frac 1 {d \cdot n+1}$.
	It can be checked by using Stirling's approximation $n! \sim \sqrt{2\pi n} \left(\frac n e\right)^n$
	that $\lim_{n \to \infty} \frac {f(n+1)} {f(n)} = \frac {d^d}{(d-1)^{d-1}} = d \cdot (1 + \frac 1 {d-1})^{d-1}$.
	The latter quantity is upper bounded by $d \cdot e$ for every $d \geq 1$.
\end{proof}

\paragraph{Generatingfunctionology.}

The \emph{formal power series} (a.k.a.~\emph{ordinary generating function})
associated with a number sequence $a : \N \to \R$
is the infinite polynomial 
%
	$g_a(x) = \sum_{n=0}^\infty a(n) \cdot x^n$. 
%
Let $\coefficient {x_n} g_a$ denote the coefficient $a(n)$ of $x^n$ in $g_a$.
%
%
%
Let $f, f_1, f_2 : \N \to \R$ be sequences.
It is well known that $g_k(x) = k$ for $k \in \R$,
$g_{f_1+f_2} = g_{f_1} + g_{f_2}$,
$g_{f_1 \conv f_2} = g_{f_1} \cdot g_{f_2}$,
and $g_f(x) = f(0) + x \cdot g_{\sigma f}(x)$.
Consequently, if $f_1$ is conv-recursive with auxiliary sequences $f_1, \dots, f_k$,
then their generating functions $g_{f_1}, \dots, g_{f_k}$
satisfy the following system of polynomial equations
\begin{align}
	\label{eq:gf}
	\left\{
		\begin{array}{rcl}
			g_{f_1}(x) &=& f_1(0) + x \cdot \hat p_1(g_{f_1}(x), \dots, g_{f_k}(x)), \\
			&\vdots& \\
			g_{f_k}(x) &=& f_k(0) + x \cdot \hat p_k(g_{f_1}(x), \dots, g_{f_k}(x)).
		\end{array}
	\right.
\end{align}
where $\hat p_i$ is the polynomial obtained from the convolution polynomial $p_i$
by replacing the convolution operation ``$\conv$'' on sequences
by the product operation ``$\cdot$'' on real numbers.
Thus, the generating function $g_f$ of a conv-rec sequence $f$ is \emph{algebraic}.

\begin{lemma}
	\label{lem:unique}
	The system of equations \eqref{eq:gf} has a unique formal power series solution.
\end{lemma}
\begin{proof}
	By construction, $g_f = (g_{f_1}, \dots, g_{f_k})$ is a formal power series solution of \eqref{eq:gf}.
	We now argue that there is no other solution.
	Assume that $g = (g_1, \dots, g_k)$ is a solution of \eqref{eq:gf}.
	We prove that, for every $n \in \N$,
	$\coefficient {x^n} g = (\coefficient {x^n} {g_1}, \dots, \coefficient {x^n} {g_1})$
	equals $\coefficient {x^n} {g_f} = (\coefficient {x^n} {g_{f_1}}, \dots, \coefficient {x^n} {g_{f_k}})$.
	The base case follows immediately from \eqref{eq:gf},
	since $\coefficient {x^0} {g_i} = f_i(0)$ by definition.
	For the inductive step $n > 0$, notice that
	\begin{inparaenum}[1)]
		\item from \eqref{eq:gf} we have
		$\coefficient {x^n} {g_i} = \coefficient {x^n} {(x\cdot\hat p_i(g))} = \coefficient {x^{n-1}} {\hat p_i(g)}$, and
		\item the latter quantity is a (polynomial) function
		of the coefficients $\coefficient {x^i} g$ for $0 \leq i \leq n-1$.
	\end{inparaenum}
	By inductive assumption,
	$\coefficient {x^i} g = \coefficient {x^i} {g_f}$ for every $0 \leq i \leq n-1$,
	and thus by the two observations above
	$\coefficient {x^n} g = \coefficient {x^n} {g_f}$.
\end{proof}

\begin{lemma}
	\label{cor:unique}
	Let $d$ be the combined degree of $f = \tuple{f_1, \dots, f_k}$.
	The system \eqref{eq:gf}
	has a unique solution
	$g_f(x^*) = \tuple {g_{f_1}(x^*), \dots, g_{f_k}(x^*)} \in \R^k$
	for every $0 \leq x^* < \frac 1 d$.
\end{lemma}

\begin{proof}
	Let $g_f = \tuple{g_{f_1}, \dots, g_{f_k}}$
	be the tuple of formal power series of the sequences $f_1, \dots, f_k$.
	By \cref{lem:bounded-ratio},
	$\lim_{n\to\infty} \frac {f_i(n+1)} {f_i(n)} = O(d)$.
	Thus, $g_f(x^*) = \tuple {g_{f_1}(x^*), \dots, g_{f_k}(x^*)} \in \R^k$
	converges for every $0 \leq x^* < \frac 1 d$.
	%
	By \cref{lem:unique},
	$g_f$ is the unique formal power series solution of \eqref{eq:gf}.
\end{proof}





\begin{theorem}
	\label{thm:zeroness}
	The zeroness problem for conv-rec sequences is in \PSPACE.
\end{theorem}

\begin{proof}
	Let $f_1$ be a conv-rec sequence of combined degree $d$
	with auxiliary sequences $f_2, \dots, f_k$
	satisfying \eqref{eq:convrec}.
	Consider the associated generating functions
	$g = \tuple{g_{f_1}, \dots, g_{f_k}}$.
	Clearly, $f_1(n) = 0$ for every $n \in \N$
	if, and only if, $g_{f_1}(x) = 0$ for every $x$ sufficiently small.
	By \cref{cor:unique},
	$g(x^*)$ is the unique solution of \eqref{eq:gf}
	for every $0 \leq x^* < \frac 1 d$.
	It thus suffices to say that,
	for every $0 \leq x^* < \frac 1 d$,
	\emph{all} solutions $g(x^*)$ of the system \eqref{eq:gf}
	satisfy	$g_{f_1}(x^*) = 0$.
	This can be expressed by the following universal first-order sentence over the reals
	(where $\bar y = \tuple{y_1, \dots, y_k}$)
	\begin{align*}
		\forall \left(0 \leq x < \frac 1 d\right) \st \forall \bar y \st \bar y = f(0) + x \cdot \hat p(\bar y) \limplies y_1 = 0.
	\end{align*}
	%
	The sentence above can be decided in \PSPACE by appealing to the existential theory of the reals
 	\cite[Theorem 3.3]{Canny:1988:STOC:1988}.
\end{proof}



\section{Universality of unambiguous grammars}


Let $\Sigma$ be a finite alphabet.
We denote by $\Sigma^*$ the set of all finite words over $\Sigma$,
including the empty word $\varepsilon$.
A \emph{language} is a subset $L \subseteq \Sigma^*$.
The concatenation of two languages $L, M \subseteq \Sigma^*$
is \emph{unambiguous} if $w \in L\cdot M$ implies that $w$ factors uniquely as $w = u \cdot v$
with $u \in L$ and $M \in v$.
A \emph{context-free grammar} (\CFG) is a tuple $G = \tuple {\Sigma, N, S, {\from}}$
where $\Sigma$ is a finite alphabet of \emph{terminal symbols},
$N$ is a finite set of \emph{nonterminal symbols},
of which $S \in N$ is the \emph{starting nonterminal symbol},
and ${\from} \subseteq N \times (N \cup \Sigma)^*$
is a set of productions.
A \CFG is in \emph{short Greibach normal form}
if productions are of the form either $X \from \varepsilon$.
or $X \from aYZ$.
An \emph{$X$-derivation tree} is a tree satisfying the following conditions:
1) the root node $\varepsilon$ is labelled by the nonterminal $X \in X$,
2) every internal node is labelled by a nonterminal from $N$,
3) whenever a node $u$ has children $u\cdot 1, \dots, u\cdot k$
there exists a rule $Y \from w_1 \cdots w_k$ with $w_i \in N \cup \Sigma$
\st $Y$ is the label of $u$
and $w_i$ is the label of $u\cdot i$, and
4) leaves are labelled with terminal symbols from $\Sigma$.
The \emph{language recognised} by a nonterminal $X$
is the set $\lang X$ of words $w = a_1 \cdots a_n \in \Sigma$
\st there exists an $X$-derivation tree with leaves labelled by (left-to-right) $a_1, \dots, a_n$;
the language recognised by $G$ is the language recognised by the starting nonterminal $\lang G = \lang S$.
A \CFG $G$ is \emph{unambiguous} (\UCFG)
if for every accepted word $w \in \lang G$
there exists exactly one derivation tree witnessing its acceptance.
The universality problem (\UUCFG) asks,
given a \UCFG $G$,
whether $\lang G = \Sigma^*$.

\subsection{Reductions}

In this section present \PTIME reductions from inclusion problems for \NFA and \UCFG to \UUCFG.
This serves us as a motivation to study the complexity of \UUCFG in \cref{sec:UUCFL}.
We proceed in two steps.
In the first step, we present a general l.h.s.~determinisation procedure for inclusion problems (\cref{sec:det:trick}) which is widely applicable to essentially any machine-based model of computation.
In the second step, assuming a deterministic l.h.s., we show a reduction from inclusion to universality (\cref{sec:incl2univ}).
We apply these two reductions in \cref{sec:applications}.

\subsubsection{L.h.s.~determinisation for inclusion problems}
\label{sec:det:trick}

It is an empirical observation that in many inclusion problems of the form $L \subseteq M$
the major source of difficulty is with $M$ and not with $L$.
For example, for finite automata the inclusion problem is \PSPACE-complete when $M$ is presented by a \NFA
and in \NLOGSPACE when it is presented by a \DFA.
In either case, it is folklore that whether $L$ is presented as a \NFA or \DFA does not matter.
A more dramatic example is given when $L$ is regular and $M$ context-free,
since the inclusion above is undecidable when $M$ is presented by a \CFG
and in \PTIME when it is presented by a $\DCFG$.

In this section we give a formal explanation of this phenomenon
by providing a generic reduction of an inclusion problem as above
to one where the l.h.s.~$L$ is a deterministic language.
The reduction will be applicable under mild assumptions which are satisfied by most machine-based models of language acceptors
such as finite automata, B\"uchi automata, context-free grammars/pushdown automata, Petri-nets, register automata, timed automata, etc.
For the language class of the r.h.s.~$M$ it suffices to have closure under inverse homomorphic images,
and for the l.h.s.~$L$ it suffices that we can rename the input symbols
read by transitions in a suitable machine model%
\footnote{The reduction applies also to undecidable instances of the language inclusion problem such as $\CFG \subseteq \DCFG$,
however in this case it is of no use since $\DCFG \subseteq \DCFG$ is known to be undecidable \cite[Theorem 10.7, Point 2]{HopcroftUllman:1979}.}.
Moreover, we argue that such transformation preserves whether $M$ is recognised by a deterministic or an unambiguous machine.

Let $\Sigma$ be a finite alphabet%
\footnote{The construction below can easily be adapted to infinite alphabets of the form $\Sigma \times \mathbb A$,
where $\Sigma$ is finite and $\mathbb A$ is an infinite set of data values \cite{BojanczykKlinLasota:LMCS:2014}.}.
Assume that $L = \lang A \subseteq \Sigma^*$
is recognised\footnote{Languages of infinite words can be handled similarly.}
by a nondeterministic machine $A$ with transitions of the form $\delta = p \goesto {a, \op} q \in \Delta_A$,
where $\op$ is an optional operation that manipulates a local data structure (a stack, queue, a tape of a Turing machine, etc...).
The construction below does not depend on what $\op$ does.
We  assume w.l.o.g.~that $A$ is \emph{total},
i.e., for every control location $p$ and input symbol $a \in \Sigma$
there exists a transition of the form $p \goesto {a, \_} \_ \in \Delta_A$.
Consider a new alphabet $\Sigma' = \Delta_A$,
together with the projection homomorphism $h : \Sigma' \to \Sigma$
that maps a transition $\delta = p \goesto {a, \op} q \in \Delta_A$
to its label $h(\delta) = a \in \Sigma$.
We modify $A$ into a new machine $A'$ by replacing each transition $\delta$ above with
$p \goesto {\delta, \op} q \in \Delta_{A'}$.
Intuitively, $A'$ behaves like $A$ except that it needs to declare which transition $\delta$
it is actually taking in order to read $a = h(\delta)$.
By construction, $A'$ is deterministic (in fact, every transition has a unique label across the entire machine)
and $\lang A = h(\lang {A'})$ is the homomorphic image of $\lang {A'}$.

We need to adapt the machine $B$ recognising $M = \lang B$ in order to preserve inclusion.
For every transition $r \goesto {a, \op} s \in \Delta_B$
and for every $\delta = p \goesto {b, \op'} q \in \Delta_A$ with $b = a$,
we have in $B'$ a transition $r \goesto {\delta, \op} s \in \Delta_{B'}$.
Intuitively, $B'$ behaves like $B$ except that it reads additional information on the transition taken by $A'$.
This information is not actually used by $B'$ during its execution
but it is merely added in order to lift the alphabet from $\Sigma$ to $\Sigma'$.
We have $\lang {B'} = h^{-1}(\lang B)$ is the inverse homomorphic image of $\lang B$.
The following lemma states the correctness of the reduction.
\begin{lemma}
	%
	We have the following equivalence:
		$\lang A \subseteq \lang B
			 \textrm{ if, and only if, }
				\lang {A'} \subseteq \lang {B'}.$
\end{lemma}
\begin{proof}
	By generic properties of images and inverse images we have the following two inclusions:
	\begin{align}
	\label{eq:prop:h}
		\lang {A'} \subseteq h^{-1}(h(\lang {A'}))
			\quad \textrm{ and } \quad
				h(h^{-1}(\lang B)) \subseteq \lang B.
	\end{align}
	For the ``only if'' direction,
	if $\lang A \subseteq \lang B$ holds,
	then $h^{-1}(\lang A) \subseteq h^{-1}(\lang B)$,
	which, by the definition of $A'$ and $B'$,
	is the same as $h^{-1}(h(\lang A)) \subseteq \lang {B'}$.
	By \eqref{eq:prop:h},
	$\lang {A'} \subseteq h^{-1}(h(A')) \subseteq \lang {B'}$, as required.
	For the ``if'' direction,
	if $\lang {A'} \subseteq \lang {B'}$ holds,
	then also $h(\lang {A'}) \subseteq h( \lang {B'})$ holds.
	Similarly as above,
	we have $\lang A = h(\lang {A'}) \subseteq h(\lang {B'}) = h(h^{-1}(\lang B)) \subseteq \lang B$,
	as required.
\end{proof}

The following lemma states that the reduction above preserves whether $B$ is deterministic or unambiguous.
We mean here the following generic semantic notion of unambiguity:
$B$ is \emph{unambiguous} if for every $w \in \Sigma^*$,
there exists at most one accepting run of $B$ over $w$.
(This notion specialises to the classical notion of unambiguity of finite automata, pushdown automata, Parikh automata, etc.)
\begin{lemma}
	If $B$ is deterministic, then so is $B'$.
	If $B$ is unambiguous, then so is $B'$.
\end{lemma}

\begin{proof}
	A transition $p \goesto {\delta, \op} q \in \Delta_{B'}$ in $B'$
	is obtained taking several distinct copies of a transition
	${p \goesto {a, \op} q \in \Delta_B}$ in $B$
	w.r.t.~every possible transition $\delta \in \Delta_A$ over the same input symbol $h(\delta) = a$.
	By way of contradiction, assume that $B$ is deterministic
	and that $B'$ is not deterministic.
	There are two distinct transitions
	$p \goesto {\delta, \op_1} q_1, p \goesto {\delta, \op_2} q_2 \in \Delta_{B'}$
	in $B'$	from the same control location $p$ and input $\delta \in \Sigma'$.
	If $\delta$ is labelled by $h(\delta) = a \in \Sigma$,
	then by construction there are two distinct transitions
	$p \goesto {a, \op_1} q_1, p \goesto {a, \op_2} q_2 \in \Delta_B$
	in $B$ over the same input symbol $a$.
	This contradicts the fact that $B$ was assumed to be deterministic,
	and thus $B'$ must be deterministic as well.
	An analogous argument shows that also unambiguity is preserved.
\end{proof}

\subsubsection{From inclusion to universality}
\label{sec:incl2univ}

Let $\mathcal L$ and $\mathcal M$ be two classes of languages
and let $L \in \mathcal L$ and $M \in \mathcal M$.
A naive approach to decide the inclusion problem
(and the most common)
is to use the following equivalence:
\begin{align}
	\label{eq:incl2ne}
	L \subseteq M
    \quad \textrm{ if, and only if, } \quad
      L \cap (\Sigma^* \setminus M) = \emptyset.
\end{align}
However, this requires complementation of $M$,
which is either expensive (exponential complexity for \NFA)
or just impossible (context-free languages are not closed under complemenetation,
even for the unambiguous subclass \cite{Hibbard:Ullian:UCFL:1966}).
However, we observe the following related reduction
which works much better in our setting:
\begin{align}
	\label{eq:incl2univ}
	L \subseteq M
		\quad \textrm{ if, and only if, } \quad
			(M \cap L) \;\cup\; (\Sigma^* \setminus L) = \Sigma^*.
\end{align}
On the face of it, this looks more complicated than \eqref{eq:incl2ne}
because we now have to perform a complementation (of $L$),
an intersection, a union,
and finally we reduce to the universality problem instead of the nonemptiness,
which is still difficult in general.
However, in our setting there are gains.
First of all, thanks to \cref{sec:det:trick}
we can assume that $L$ is a deterministic language,
and thus complementation is usually available (and cheap).
Second, while universality is still a difficult problem,
it can be easier than inclusion,
e.g., \DCFG inclusion is undecidable while \DCFG universality is decidable (even in \PTIME).

In order to apply \eqref{eq:incl2univ}
we require that $\mathcal L$ is a deterministic class efficiently closed under complement
(i.e., a representation for the complement is constructible in \PTIME)
and that the class $\mathcal M$ is closed under disjoint unions and intersections with languages from $\mathcal L$.
Most deterministic languages classes, such as those recognised by deterministic finte automata,
deterministic context-free grammars,
deterministic Parikh automata,
deterministic register automata, etc.,
satisfy the first requirement%
\footnote{
	A notable exception is deterministic Petri-net languages under coverability semantics,
	since the complement of such languages intuitively requires checking whether some counter is negative,
	which is impossible without zero tests.
	In fact, if both a language and its complement are deterministic Petri-net recognisable under coverability semantics,
	then they are both regular \cite{CzerwinskiLasotaMeyerMuskallaKumarSaivasan:CONCUR:2018}.
}.
The second requirement is satisfied for classes of languages for which the underlying machine models admit a product construction%
\footnote{
	As an example not satisfying this requirement,
	one can take $\mathcal L = \mathcal M$
	to be the class of \DCFL,
	since they are not closed under intersection.
	In fact, while we show in this paper that \UUCFG is decidable,
	the equivalence problem for \UCFG is open.
}.

\subsubsection{Applications}
\label{sec:applications}

In this section we apply the reductions of \cref{sec:det:trick} and \cref{sec:incl2univ}
in order to reduce certain inclusion problems to their respective universality variant.

\begin{theorem}
    \label{thm:NFA:UFA}
    ``$\NFA \subseteq \UFA$'' is in \PTIME.
\end{theorem}
\noindent
While equivalence and inclusion of \UFA is well-known to be in \PTIME \cite[Corollary 4.7]{StearnsHunt:Unambiguous},
the same complexity for the more general problem ``$\NFA \subseteq \UFA$'' does not seem to have been observed before.
\begin{proof}
    By \cref{sec:det:trick},
    the problem reduces to ``$\DFA \subseteq \UFA$''.
    By \eqref{eq:incl2univ},
    $L \subseteq M$ is equivalent to $N := M \cap L \cup (\Sigma^* \setminus L) = \Sigma^*$.
    Notice that $N$ is effectively \UFA,
    since the \DFA language $L$ can be complemented in \PTIME,
    the intersection $M \cap L$ is also \UFA and computable in quadratic time,
    and the disjoint union of a \UFA and a \DFA is also a \UFA computable in linear time.
    Since the universality problem for unambiguous automata can be solved in \PTIME,
    also ``$\DFA \subseteq \UFA$'', and thus ``$\NFA \subseteq \UFA$'', is in \PTIME as well.
\end{proof}

\begin{theorem}
    \label{thm:NFA:UCFG}
    ``$\NFA \subseteq \UCFG$'' is \PTIME inter-reducible with \UUCFG.
\end{theorem}

\begin{proof}
    By \cref{sec:det:trick},
    the problem reduces to ``$\DFA \subseteq \UCFG$''.
    Thanks to \cref{sec:incl2univ},
    the latter problem reduces to \UUCFG
    since 1) \DFA languages are efficiently closed under complement (in \PTIME),
    2) \UCFG languages are efficiently closed under intersection with \DFA languages (in \PTIME), and
    3) the disjoint union of a \UCFG language and a \DFA language is a \UCFG language.
    %
    %
    %
    %
    %
    Thus, ``$\NFA \subseteq \UCFG$'' reduces to \UUCFG,
    and since \UUCFG is a special case of the former problem,
    ``$\NFA \subseteq \UCFG$'' is \PTIME inter-reducible with \UUCFG.
\end{proof}

\begin{theorem}
    \label{thm:CFG:UFA}
    ``$\CFG \subseteq \UFA$'' reduces to \UUCFG.
\end{theorem}
\begin{proof}    
    By \cref{sec:det:trick}, ``$\CFG \subseteq \UFA$'' reduces to ``$\DCFG \subseteq \UFA$'',
    which in turn reduces to $\UUCFL$ thanks to \cref{sec:incl2univ}
    because 1) $\DCFG$ languages are efficiently closed under complement,
    2) the intersection of a \UFA and a \DCFG language is efficiently \DCFG, and
    3) the disjoint union of two \DCFG languages is efficiently \UCFG. 
    %
    %
    (The latter problem reduces to universality of two disjoint \DCFG languages,
    which in principle may be easier than \UUCFL.)
\end{proof}

\subsection{\UUCFG in \PSPACE}
\label{sec:UUCFL}

In this section we show that \UUCFG is in \PSPACE
by reducing to the zeroness problem for conv-rec sequences.
This complexity upper bound appears also in \cite{Purgal:MSC:2018},
albeit with a more direct argument reducing to systems of monotone polynomial equations.


Let $\Sigma = \set{a, b}$ be a finite alphabet
and let $L \subseteq \Sigma^*$ be a language of finite words over $\Sigma$.
The \emph{counting function} of $L$ is the sequence $f_L : \N \to \N$
s.t.~for every $n \in \N$,
$f_L(n) = \card {L \cap \Sigma^n}$ counts the number of words of length $n$ in $L$.
Given a unambiguous context-free grammar $G = \tuple{\Sigma, N, S, \from}$ in short Greibach normal form,
let $f_X := f_{\lang X} : \N \to \N$ be the counting function
of the language $\lang X$ recognised by the nonterminal $X \in N$.
It is well-known that the $f_X$'s satisfy the following system of equations with convolution:
\begin{align}
	\label{eq:count}
	f_X(n+1) = 
	    \sum_{X \from a Y Z} (f_Y \conv f_Z)(n).
\end{align}
The initial condition is $f_X(0) = 1$ if $X \from \varepsilon$ and $f_X(0) = 0$ otherwise.
In other words, $f_S$,
which is the counting function of the language $\lang G$ recognised by $G$,
is conv-rec.
Unambiguity is used crucially to show that any word $w$ in $\lang {Y \cdot Z}$
factorises uniquely as $w = u\cdot v$ with $u \in \lang Y$ and $v \in \lang Z$,
which allows us to obtain $f_{\lang {Y \cdot Z}} = f_{\lang Y} \conv f_{\lang Z}$.

Clearly, $G$ is universal if, and only if,
$f_S$ is identically equal to the sequence $g(n) = 2^n$.
The latter sequence is conv-rec since it satisfies
$g(n+1) = (g \conv g)(n)$,
with the initial condition $g(0) = 1$.
Thus $G$ is universal if, and only if,
$f(n) = g(n) - f_S(n)$ is identically zero.
%
Since conv-rec sequences are closed under subtraction, 
$f(n)$ is also conv-rec.
By \cref{thm:zeroness}, we can decide zeroness of $f$ in \PSPACE,
and thus the same upper bound holds for \UUCFG.

\begin{theorem}\label{thm:UUCFL:PSPACE}
	The universality problem for unambiguous context-free grammars \UUCFG is in \PSPACE.
\end{theorem}

\section{\SQRTSUM-hardness of coin-flip measure}
\label{sec:coin-flip}

In this section we show that a quantitative generalisation of \UUCFG is hard
for a well-known problem in numerical computing.
Let $\Sigma_n = \set{a_1, \dots, a_n}$ be a finite alphabet of $n$ distinct letters.
Consider the following random process to generate a finite word in $\Sigma^*$.
%
At step $k$ we select one option $a_k \in \Sigma_\varepsilon = \Sigma_n \cup \set \varepsilon$
uniformly at random.
If $a_k = \varepsilon$,
then we terminate and we produce in output $a_0 \cdots a_{k-1}$.
Otherwise, we continue to the next step $k+1$.
It is easy to see that the probability to generate a word depends only on its length
and equals $\mu_\text{coin}(w) = \left(\frac 1 {\card \Sigma+1}\right)^{\card w +1}$.
The \emph{coin-flip measure} of a language of finite words $L\subseteq \Sigma^*$ is
%
	$\mu_\text{coin}(L) = \sum_{w \in L} \mu_\text{coin}(w).$
%
Clearly, $0 \leq \mu_\text{coin}(L) \leq 1$,
$\mu_\text{coin}(L) = 0$ iff $L = \emptyset$,
and $\mu_\text{coin}(L) = 1$ iff $L = \Sigma^*$.

Since $\mu_\text{coin}(w)$ depends just on $\abs w$,
we can write $\mu_\text{coin}(L) = \sum_{k=0}^\infty f_L(k) \cdot \left(\frac 1 {n+1}\right)^{k+1}$,
where $f_L(k) = \abs {L \cap \Sigma^k}$ is the counting function of $L$.
In other words, one possible way of computing the coin-flip measure
it by evaluating the generating function $g_{f_L}(x)$
at $x = \frac 1 {n+1}$ (modulo a correction factor):
$\mu_\text{coin}(L) = \frac 1 {n+1} \cdot g_{f_L}\left(\frac 1 {n+1}\right)$.
Consequently, the coin-flip measure of a regular language is rational,
and that of an unambiguous context-free language is algebraic
(following from the analogous, and more general,
facts about the respective generating functions
\cite{ChomskySchutzenberger:Algebraic:1963}).
Let $L, M \subseteq \Sigma_n^*$ be two languages with unambiguous concatenation $L \cdot M$.
Then
\begin{align}
    \label{eq:mu:LM}
    \mu_\text{coin}(L \cdot M) = (n+1) \cdot \mu_\text{coin}(L) \cdot \mu (M).
\end{align}

The \emph{coin-flip comparison problem} asks,
given a language $L \subseteq \Sigma^*$,
a rational threshold $0 \leq \varepsilon \leq 1$ encoded in binary,
and a comparison operator ${\sim} \in \set{\leq, <, >, \geq}$,
whether $\mu_\text{coin}(L) \sim \varepsilon$ holds.
The universality problem for $L$ is the special case
when $\varepsilon = 1$.
%
%
We now relate the coin-flip comparison problem to an open problem in numerical computing.
The \emph{\SQRTSUM problem} asks,
given $d_0, \dots, d_n \in \N$ encoded in binary
and a comparison operator ${\sim} \in \set{{\leq}, {<}, {>}, {\geq}}$,
whether
\footnote{
	In fact, the problem reduces to the case when ${\sim} = {\geq}$ is fixed.
	By doing binary search in the interval $\set{0, 1, \dots, n \cdot d}$,
	with only $O(\log (n\cdot d))$ queries to \eqref{eq:SQRTSUM}
	we can find the unique $\hat d_0 \in \N$
	\st $\hat d_0 \leq \sum_{i=1}^n \sqrt {d_i} \leq \hat d_0 + 1$.
	We can then solve $\sum_{i=1}^n \sqrt {d_i} \leq d_0$
	by checking $d_0 \leq \hat d_0 + 1$,
	and similarly for the other comparison operators.
}:
\begin{align}
	\label{eq:SQRTSUM}
	\sum_{i=1}^n \sqrt {d_i} \sim d_0.
\end{align}
This problem can be shown to be in \PSPACE
by deciding the existential formula
$\exists x_1, \dots, x_n \st x_1^2 = d_1 \land \cdots \land x_n^2 = d_n \land x_1 + \cdots + x_n \sim d_0$
over the reals \cite{Canny:1988:STOC:1988}.
It is a long-standing open problem in the theory of numerical computation
whether \SQRTSUM is in \NP,
or whether it is \NP-hard
\cite{AllenderBurgisserKjeldgaard-PedersenMiltersen:JC:2009,EtessamiYannakakis:JACM:2009}.

\begin{theorem}
    \label{thm:SQRTSUM-hardness}
    The coin-flip measure comparison problem is \SQRTSUM-hard for \UCFG.
\end{theorem}





In the rest of the section we prove the theorem above.
Let $d_0, \dots, d_n \in \N$ be the input to \SQRTSUM.
We assume w.l.o.g.~that $n$ is an odd number $\geq 3$.
We construct a rational constant $\varepsilon \in \Q$
and a \UCFG ${G = \tuple{\Sigma_n, N, X_0, \from}}$
over a $n$-ary alphabet $\Sigma_n = \set {a_1, \dots, a_n}$
and nonterminals $N$ containing $\set{X_0, \dots, X_n, C_1, \dots, C_n, A}$
plus some auxiliary nonterminals
(omitted for readability)
\st $\mu_\text{coin}(\lang G) \sim \varepsilon$ if, and only if, \eqref{eq:SQRTSUM} holds.
The principal productions of the grammar are:
\begin{align*}
	X_0 &\from a_1 \cdot X_1 \sep \cdots \sep a_n \cdot X_n, \\
    X_1 &\from C_1 \sep A \cdot X_1 \cdot a_n \cdot X_1, \\
        &\vdots \\
    X_n &\from C_n \sep A \cdot X_n \cdot a_n \cdot X_n.
\end{align*}
The remaining nonterminals $C_i$'s and $A$ will generate certain regular languages to be determined below.
Let $d = \max_{i = 1}^n d_i$.
For every $1 \leq i \leq n$,
let $x_i = 1 - \frac {\sqrt {d_i}} d$. 
It is easy to check that $x_i$
is the least non-negative solution of
\begin{align}
	\label{eq:xi}
	x_i	&= c_i + a \cdot x_i^2
	    \quad \text{where } c_i := \frac 1 2 \left(1 - \frac {d_i} {d^2}\right)
	    \text{ and } a := \frac 1 2 .
\end{align}
In the following, we write $\mu(X)$ for a non-terminal $X \in N$
as a shorthand for $\mu_\text{coin}(\lang X)$.
Since $\mu(a_1) = \cdots = \mu(a_n) = \frac 1 {(n+1)^2}$,
by \eqref{eq:mu:LM} we have
\begin{align}
    \label{eq:Xi}
    \mu(X_0) = \frac 1 {n+1} (\mu(X_1) + \cdots + \mu(X_n))
    \text{ and }
    \mu(X_i) = \mu(C_i) + (n+1) \cdot \mu(A) \cdot \mu(X_i)^2, i \in \set{1, \dots, n}.
\end{align}
We aim at obtaining $\mu(X_i) = x_i$.
By comparing \eqref{eq:Xi} with \eqref{eq:xi}
we deduce that the nonterminals $C_i$ and $A$
must generate languages of measure $\mu(C_i) = c_i$, resp., $\mu(A) = \frac a {n+1} = \frac 1 {2(n+1)}$.
Since the measures $a, c_i$ are rational,
it suffices to find regular languages $\lang A, \lang {C_i}$.
The main difficulty is to define these language as to ensure that $G$ is unambiguous
and of polynomial size.
In order to achieve this we further require that
\begin{inparaenum}[1)]
    \item $\lang A \subseteq \Sigma_{n-1}$ is a finite set of words of length $1$ (single letters)
    not containing letter $a_n$, and
    \item $\lang {C_i} \subseteq \Sigma_{n-1}^*$ is a set of words not containing letter $a_n$.
\end{inparaenum}
%

We first define $\lang A$.
Let
\begin{align*}
    A \from a_1 \sep \cdots \sep a_{\frac {n+1} 2}.    
\end{align*}
In order to avoid letter $a_n$,
we require $\frac {(n+1)} 2 \leq n-1$.
The latter condition is satisfied since we assumed $n \geq 3$.
Thus, $\lang A \subseteq \Sigma_{n-1}$ is finite,
contains only words of length $1$,
and has measure $\mu (A) = \frac {n+1} 2 \cdot \frac 1 {(n+1)^2} = \frac a {n+1}$, as required.

The definition of $\lang {C_i}$ of measure $\mu (C_i) = c_i$ is more involved.
In general, it is easy to construct a regular expression (or a finite automaton) recognising a language
of measure equal to a given rational number.
However, we have two constraints to respect:
\begin{inparaenum}[1)]
    \item we can use only letters from $\Sigma_{n-1}$, and
    \item the regular expression must have size polynomial in the bit encoding of $c_i$.
\end{inparaenum}
The first constraint entails an upper bound $\mu(\Sigma_{n-1}^*) = \frac 1 2$
on the maximal measure that $\lang {C_i}$ can have.
However, this is not a problem in our case since $c_i < \frac 1 2$ by definition.
The second constraint is handled by the following lemma.
%
\begin{restatable}[Representation lemma]{lemma}{lemmaRepresentation}
	\label{lem:rep}
	Let $n + 1 \in \N$ with $n \geq 2$ be a base,
	let $m \in \N$ \st $1 \leq m \leq n$,
	and let $c \in \R$ with $0 \leq c \leq \frac 1 {n - m + 1}$ be a target rational measure
	written in reduced form as
		$c = \frac p q, \textrm{ with } p, q \in \N,\ p \leq q$.
	%
	There exists an unambiguous regular expression $e$
	using only letters from $\Sigma_m \subseteq \Sigma_n$ recognising a language of measure
    $\mu(\lang e) = c$.
	Moreover, if there exists $\ell \in \N$ \st $q \divides (n+1)^\ell$,
	then $e$ can be taken of size polynomial in $\log q$, $n$, and $\ell$.
\end{restatable}
We apply Lemma~\ref{lem:rep} with $m := n-1$ and $c := c_i$
and obtain an unambiguous regular expression $e$
recognising a language $\lang e \subseteq \Sigma_{n-1}^*$ of measure $c_i$.
We now argue that $e$ can be taken of polynomial size.
In order to achieve this, we assume w.l.o.g.~that $d = (n+1)^{2h}$ for some $h$.
(This can be ensured by adding a new integer 
$d_{n+1} = (n+2)^{2h}$ for some $h$ large enough,
and by replacing $d_0$ with $d_0 + \sqrt{d_{n+1}} = d_0 + (n+2)^h$.)
Consequently, $c_i = \frac {\frac d 2 (d^2 - d_i)} {(n+1)^{4h}} = \frac p q$
with $p, q \in \N$ relatively prime and $q \divides (n+1)^{4h}$,
and thus $e$ has polynomial size by taking $\ell = 4h$ in the lemma.
%
%
The set of polynomially many production rules for nonterminal $C_i$
is derived immediately from the regular expression $e$
by adding some auxiliary nonterminals.
Moreover, since $e$ is unambiguous, the same applies to the rules for $C_i$.
This completes the description of the grammar $G$.

\begin{lemma}
    The grammar $G$ is unambiguous.
\end{lemma}

\begin{proof}
    Since $\lang {G} = \lang {X_0}$
    is the union of languages $\lang{a_1 \cdot X_1}, \dots, \lang{a_n \cdot X_n}$,
    and the latter are disjoint,
    it suffices to show that the $\lang{X_i}$'s are recognised unambiguously.
    Let $w \in \lang {X_i}$.
    If $w$ does not contain any $a_n$, then necessarily $w \in \lang {C_i}$.
    Otherwise, let $w = u a_n v$ where $v$ does not contain any $a_n$.
    Thus $v \in \lang {C_i}$ and $u \in \lang {A \cdot X_i}$.
    Since $A$ produces only words of fixed length,
    $u = xw'$ unambiguously with $x \in A$ and $w' \in \lang {X_i}$.
    This argument shows that for any $w \in \lang{X_i}$
    if we let $s$ be the number of $a_n$ in $w$,
    then $w \in \lang{A^s \cdot C_i \cdot (a_n \cdot C_i)^s}$.
    Since $A$ produces words of fixed length and $C_i$ does not produce any word containing $a_n$,
    the latter concatenation is unambiguous and thus $w$ is produced unambiguously by $X_i$.
    %
\end{proof}

Let $\varepsilon := \frac 1 {n+1} \left(n - \frac {d_0} d \right)$.
The following lemma states the correctness of the reduction.
\begin{lemma}
    We have $\mu(\lang G) \sim \varepsilon$ if, and only if, \eqref{eq:SQRTSUM} holds.
\end{lemma}

\begin{proof}
    Since $x_i = 1 - \frac {\sqrt {d_i}} d$,
    we have 
    %
    	$\mu(X_0) = 
    		\mu(a_1 \cdot X_1) + \cdots + \mu(a_n \cdot X_n) 
    		= (n+1) (\mu(a_1) \cdot \mu (X_1) + \cdots + \mu(a_n) \cdot \mu(X_n))
    		= \frac 1 {n+1} (\mu(X_1) + \cdots + \mu(X_n))
    		= \frac 1 {n+1} (x_1 + \cdots + x_n)
    		= \frac 1 {n+1} \left(\left(1 - \frac {\sqrt {d_1}} d\right) + \cdots + \left(1 - \frac {\sqrt {d_n}} d\right)\right)
    		= \frac 1 {n+1} \left(n - \frac {\sqrt{d_1} + \cdots + \sqrt{d_n}} d \right)$,
    %
    and thus $\sum_{i=1}^n\sqrt{d_i} \sim d_0$ if, and only if, $\mu(\lang X) \sim \varepsilon$,
    as required.
\end{proof}
%

\section{Discussion}

We have shown novel \PSPACE upper bounds
for several inclusion problems on \UCFG and finite automata.
We did not address language equivalence problems $L = M$,
which in principle can be easier to decide than the corresponding inclusions.
For instance, while $\DCFG \subseteq \DCFG$ is undecidable \cite[Theorem 10.7, Point 2]{HopcroftUllman:1979},
$\DCFG = \DCFG$ is decidable by the result of G. Sénizergues \cite{Senizergues:ICALP:1997}.
It is worth remarking that decidability of the equivalence problem $\UCFG = \UCFG$ is not known.
In fact, this is a special case of the \emph{multiplicity equivalence problem} for \CFG,
which asks whether two \CFGs have the same number of derivations for every word they accept.
Decidability of the latter problem is open as well \cite{Kuich:multiplicity:1994}
and inter-reducible with the language equivalence for probabilistic pushdown automata
\cite{ForejtJancarKieferWorrell:IC:2014}.
The restriction of the $\UCFG = \UCFG$ equivalence problem
to words of a given length has been studied in \cite{Litow:AC:1996}.

\paragraph{Number sequences and the zeroness problem.}

We obtained the \PSPACE upper bound for \UUCFG by reducing to the zeroness problem for conv-rec sequences.
Conv-rec sequences generalise linear difference recurrence with constant coefficients
(a.k.a.~\emph{constant-recursive} or \emph{C-finite} \cite{KauersPaule:C-finite:2011};
c.f.~also \cite{BarloyFijalkowLhoteMazowiecki:CSL:2020} and citations therein)
by allowing the convolution product in the recurrence.
They are a special case of more expressive classes
such as P-recursive \cite[Ch.~7]{KauersPaule:C-finite:2011} (a.k.a.~\emph{holonomic})
and polynomial recursive sequences \cite{CadilhacMazowieckiPapermanPilipczukSenizergues:LICS:2020}.
%
%
The zeroness problem for P-recursive sequences is decidable \cite{Zeilberger:JCAM:1990}
and the same holds for polynomial recursive sequences
(as a corollary of the existence of cancelling polynomials \cite[Theorem 11]{CadilhacMazowieckiPapermanPilipczukSenizergues:LICS:2020}).
However, no complexity upper bounds are known for those more general classes.


\paragraph{Coin-flip measure.}

As a complement to the \PSPACE upper bound for \UUCFG,
we have shown that the coin-flip measure comparison problem
$\mu_\text{coin}(\lang G) \sim \varepsilon$ of a \UCFG $G$
with ${\sim} \in \set{{\leq}, {<}, {\geq}, {>}}$ and $0 \leq \varepsilon \leq 1$
is \SQRTSUM-hard.
The main difficulty is that the measure is generated according to a fixed stochastic process.
If we relax this constraint and generate the measure according to an arbitrary finite Markov process,
then one can obtain \SQRTSUM-hardness already for \DCFG.

It is known that the quantitative decision problem for $\mu_G(\Sigma^*)$
where $G$ is a stochastic context-free grammar (\SCFG)
is \SQRTSUM-hard \cite{EtessamiYannakakis:JACM:2009}. 
Our setting is incomparable:
On the one hand we fix a particular measure,
namely the coin-flip measure $\mu_{\text{coin}}$
(which corresponds to a fixed \SCFG with rules
$X \from \varepsilon \sep a_1 \cdot X \sep \cdots \sep a_n \cdot X$).
On the other hand, we are interested in the quantity $\mu_\text{coin}(\lang G)$
where $G$ is an arbitrary \UCFG
(and thus not necessarily universal).
%

We leave it as an open problem to establish the exact complexity of the universality problem for \UCFG
and the coin-flip measure 1 problem.
When the system of polynomial equations obtained from the grammar is \emph{probabilistic} (PPS%
\footnote{The sum of all coefficients is at most 1.})
the measure 1 problem is in \PTIME \cite{EtessamiYannakakis:JACM:2009}
(and even in strongly polynomial time \cite{EsparzaGaiserKiefer:STACS:2010}).
However, the equations obtained from \UCFG are monotone (MPS) but not PPS in general.
As an example, consider a singleton alphabet $\Sigma = \set a$
and productions of the form
$X_0 \from a$ and, for $n \geq 0$,
$X_{n+1} \from X_n \cdot X_n$.
The corresponding MPS system is $x_0 = \frac 1 {2^2}$
and $x_{n+1} = 2 \cdot x_n^2$.
The former system is not a PPS, since in the second equation the coefficients sum up to $2$.
It may be argued that by the change of variable $z_n := 2 \cdot x_n$
we obtain the system $z_0 = \frac 1 2$ and $z_{n+1} = z_n^2$
which is PPS.
However, this transformation reduces the value 1 problem on the original MPS
to the value $1/2$ problem in the new PPS,
and the latter problem is not known to be in \PTIME.

One source of difficulty in the \UUCFG problem is that witnesses of non-universality can have exponential length.
Extending the previous example,
consider the additional rules $Y_0 \from \varepsilon$
and $Y_{n+1} \from Y_n \sep X_n \cdot Y_n$.
The nonterminal $X_n$ generates a single word $\lang {X_n} = \set{a^{2^n}}$ of length $2^n$.
It can be verified by induction that $Y_n$ generates all words
$\lang{Y_n} = \set{a^0, a^1, \dots, a^{2^n - 1}}$
of length $\leq 2^n - 1$,
and consequently the grammar is unambiguous.
Thus $\lang {Y_n}$ is not universal,
however the shortest witness has length $2^n$.
In terms of measures,
$\mu_\text{coin}(\lang {X_n}) = \frac 1 {2^{2^n+1}}$
and $\mu_\text{coin}(\lang {Y_n}) = 1 - \frac 1 {2^{2^n+1}}$,
and thus \UCFG have measures that can be exponentially close to 0, resp., to 1.
%
Since a word of length $n$ over a unary alphabet has measure $\frac 1 {2^{n+1}} = 2^{-O(n)}$,
if language $L$ is not universal $\mu_\text{coin}(L) < 1$,
then there is a non-universality witness of length at most $\log (1 - \mu_\text{coin}(L))$.
Thus upper bounds on $1-\mu_\text{coin}(L)$ yield upper bounds on the shortest non-universality witness.

\paragraph{The ``$\CFG \subseteq \UFA$'' problem.}

We have shown that $\CFG \subseteq \UFA$ reduces to
$\DCFG \subseteq \UFA$ and, in turn,
the latter reduces to \UUCFG
and thus can be solved in \PSPACE.
This needs not be optimal and there are reasons to suspect that better algorithms may be obtained.
If we interpret a \DCFG $G$ as a stochastic context-free grammar (\SCFG),
then query $\lang G \subseteq \lang A$ is equivalent to $\mu_G(\lang A) = 1$ when $A$ is unambiguous,
where $\mu_G$ is the measure generated by $G$ (a generalisation of the coin-flip measure).
When $A$ is \DFA, $\mu_G(\lang A)$ can be approximated in \PTIME \cite{EtessamiStewartYannakakis:ICALP:2013}.
Generalising this result for $A$ being \UFA would put $\DCFG \subseteq \UFA$ in \PTIME.

\paragraph{The regularity problem for \UCFG.}

There are other problems which are known to be undecidable for \CFG but decidable for \DCFG,
such as the regularity problem \cite{Stearns:Regularity:DPDA:IC:1967,Valiant:Regularity:DPDA:JACM:1975,Shankar:Regularity:DPDA:TCS:1991}.
An interesting open problem \cite{DiekertKopecki:CIAA:2010}
is whether the regularity problem is decidable for \UCFG.

\paragraph{Acknowledgements.}

I warmly thank Alberto Pettorossi
for his encouragement and guidance during my first steps in doing research.
I also thank an anonymous reviewer for his helpful comments on a preliminary version of this draft.

\bibliographystyle{eptcs}
\bibliography{bibliography}

\appendix
\section{Appendix}

\lemmaRepresentation*

\begin{proof}
    Fix an alphabet $\Sigma_n$ and let $m \leq n$ as in the statement of the lemma.
    If $m = n$ then there is no difficulty
    since we can just look at the periodic expansion
    $c = \sum_{i=0} \frac {c_i} {(n+1)^{i+1}}$ of $c$ in base $n+1$,
    where crucially $0 \leq c_i \leq n$,
    and one can just interpret $c_i$ as the number of words of length $i$ that
    the regular expression must accept
    (since there are $n^i$ such words of length $i$, this can always be done,
    perhaps with the exception of $i = 0$).
    However, we work under the more general constraint $m \leq n$,
    which requires a refinement of the analysis above.
    
    There are $m^k$ words of length $k$ over $\Sigma_m \subseteq \Sigma_n$,
    and thus the measure $\mu_\text{coin}(\Sigma_m^*)$ (which is always computed w.r.t.~alphabet $\Sigma_n$) satisfies
    \begin{align}
        \label{eq:mu:Sigma:m}
        \mu_\text{coin}(\Sigma_m^*)
            = \sum_{k=0}^\infty \frac {m^k} {(n+1)^{k+1}}
            = \frac 1 {n+1} \cdot \sum_{k=0}^\infty \left(\frac m {n+1}\right)^k
            = \frac 1 {n+1} \cdot \frac 1 {1 - \frac m {n+1}}
            = \frac 1 {n - m + 1}.
    \end{align}
    We consider two cases.
	If $c =  \frac 1 {n - m + 1}$,
	then just let $e = \Sigma_m^*$,
	and we are done.
	Otherwise, assume $c < \frac 1 {n - m + 1}$.
	Our aim is to find some $k \in \N$ and write $c$ as
	\begin{align}
	    \label{eq:c}
		c = \sum_{i=0}^k \frac {c_i} {(n+1)^{i+1}} +
		    \sum_{j=1}^\infty \frac {d_j} {(n+1)^{k + j + 1}},
		\quad \textrm{ with } 0 \leq c_i \leq m^i \text { and }
		    0 \leq d_j \leq n.
	\end{align}
	Intuitively, $c_i$ counts the number of words of length $i$ in $\lang e$
	for $0 \leq i \leq k$,
	and $d_j$ counts the number of words of length $k+j$ in $\lang e$
	for $j \geq 1$.
	Moreover, in the first part the $c_i$'s are allowed to be very big (up to $m^i$)
	and in the second part the $d_j$'s must be small (up to $n$).
	Thus we will be able to interpret the $d_j$'s as digits in the base $n+1$ expansion of the second quantity.
    %
    %
	We now find the required value for $k$.
	Generalising \eqref{eq:mu:Sigma:m},
	the measure of all words of length at most $k \in \N$ over the sub-alphabet $\Sigma_m$ is
	(always computed w.r.t.~$\Sigma_n$)
	\begin{align}
		\label{eq:low}
		\mu_\text{coin}(\Sigma_m^{\leq k}) = \sum_{i = 0}^k \frac {m^i} {(n+1)^{i+1}} 
		= \frac 1 {n - m + 1} \left(1 - \left(\frac m {n+1}\right)^{k+1} \right).
	\end{align}
	%
	Since $c < \frac 1 {n - m + 1} = \lim_{k \to \infty} \mu_\text{coin}(\Sigma_m^{\leq k})$,
	there is some $k$ \st $c < \mu_\text{coin}(\Sigma_m^{\leq k})$.
	Let $k$ be such a minimal number, i.e.,
	$\mu_\text{coin}(\Sigma_m^{\leq k-1}) \leq c < \mu_\text{coin}(\Sigma_m^{\leq k})$.
	Consequently, for this choice of $k$ we have
	$c_0 = m^0, c_1 = m^1, \dots, c_{k - 1} = m^{k - 1}$,
	and	$0 \leq c_k < m^k$ where
	\begin{align*}
		\mu_{\Sigma_n}(\Sigma_m^{\leq k-1}) + \frac {c_k} {(n+1)^{k+1}}
			\leq c <
				\mu_{\Sigma_n}(\Sigma_m^{\leq k-1}) + \frac {c_k + 1} {(n+1)^{k+1}}.
	\end{align*}
	For complexity considerations,
	we note that $k$ satisfies $c < \mu_\text{coin}(\Sigma_m^{\leq k})$,
	and thus 
	$k > \frac {\log(1-(n-m+1)c)} {\log m - \log (n+1)} - 1 = \frac {-\log(1-(n-m+1)c)} {\log (n+1) - \log m}$.
	The minimal denominator is achieved by $m = n$,
	and the maximal numerator by $m = 1$.
	Performing this substitution,
	we can see that one can take $k > \frac {-\log(1-nc)} {\log (n+1) - \log n}$.
	Since $-\log(1-n c) = O(\log q)$
	and $\log (n+1) - \log n = \log \frac {n+1} n = \log (1 + \frac 1 n) = O(\frac 1 n)$,
	we obtain
	\begin{align}
		\label{eq:k}
		k = O(n \log q).
	\end{align}
	This choice of $k$ also guarantees $n \leq m^k$.
	Let $d := c - \left(\mu_{\Sigma_n}(\Sigma_m^{\leq k-1}) + \frac {c_k} {(n+1)^{k+1}}\right)$
	be the remaining part of $c$ after all terms up to $k$ have been considered:
	Thus, $d$ is small in the sense that
	$0 \leq d < \frac 1 {(n+1)^{k+1}}$.
	Since $\sum_{j=1}^\infty \frac n {(n+1)^{j+1}} = \frac 1 {n+1}$,
	we can write $d$ in base $n+1$ as required by \eqref{eq:c}:
	\begin{align*}
        d = \frac 1 {(n+1)^{k}} \sum_{j=1}^\infty \frac {d_j} {(n+1)^{j+1}}, \quad \textrm{ with } 0 \leq d_j \leq n \leq m^k.
	\end{align*}
	The coefficients are small numbers $d_j$
    at most $n$, which can thus be interpreted as digits in base $n+1$.
	Since $d$ is a rational number,
	the sequence $d_1, d_2, \dots$ is ultimately periodic \cite{HardyWright:Numbers:2008}, \ie,
	there exists a threshold $j_1 \in \N$ and a period $l \in \N$
	\st, for every $j \geq j_1$,
	%
		$d_{j+l} = d_j.$
	%
	Let $\gamma_1 = d_{j_1}, \gamma_2 = d_{j_1 + 1}, \dots, \gamma_l = d_{j_1 + l - 1}$
	be the coefficients comprising the period.
	Consequently, 
	\begin{align}
		\nonumber
		d
			&= \frac 1 {(n+1)^{k}} \left(
				\sum_{j=1}^{j_1-1} \frac {d_j} {(n+1)^{j+1}} +
				\sum_{s=0}^\infty \left(\frac {\gamma_1} {(n+1)^{j_1+sl+1}} + \cdots + \frac {\gamma_l} {(n+1)^{j_1 + sl + l}}  \right)
				\right) \\
			&= \frac 1 {(n+1)^{k}} \left(
				\sum_{j=1}^{j_1-1} \frac {d_j} {(n+1)^{j+1}} +
				\frac 1 {(n+1)^{j_1-1}} \sum_{s=0}^\infty \frac {\gamma} {(n+1)^{(s+1)l+1}}
				\right) = \\
	\label{eq:beta:next}
			&= \frac 1 {(n+1)^{k}} \left(
				\sum_{j=1}^{j_1-1} \frac {d_j} {(n+1)^{j+1}} +
				\frac \gamma {(n+1)^{j_1}} \frac 1 {(n+1)^l - 1}
				\right), 
	\end{align}
	where $\gamma := \gamma_1 (n+1)^{l-1} + \cdots + \gamma_l (n+1)^0$
	is the number of words of length $k+j_1-1+(s+1)l$, 
	for every $s \in \N$.

	We are now ready to build the regular expression $e$.
	We begin with a basic building block.
	For a given length $k$ and a cardinality $0 \leq h \leq m^k$
	written in base $m$ as $h = \sum_{i=0}^k h_i \cdot m^i$ ($0 \leq h_i < m$),
	consider the regular expression $e_{h, k}$
	\st $e_{h, k} = \Sigma_m^k$ if $h = m^k$, and otherwise:
	\begin{align*}
		e_{h, k} =
			a_m^{k-1} \cdot (a_1 + \cdots + a_{h_0}) \cdot \Sigma_m^0
		+	a_m^{k-2} \cdot (a_1 + \cdots + a_{h_1}) \cdot \Sigma_m^1 + \cdots
		+	a_m^0 \cdot (a_1 + \cdots + a_{h_{k-1}}) \cdot \Sigma_m^{k-1}.
	\end{align*}%
	\begin{claim}
	    \label{claim:ehk}
	    The expression $e_{h, k}$ is unambiguous
	    and $\lang{e_{h, k}}$ contains precisely $h$ words of length $k$.
	    Consequently, $\mu_\text{coin}(\lang{e_{h, k}}) = \frac h {(n+1)^{k+1}}$.
	\end{claim}
	\begin{proof}[Proof (of the claim)]
        All words recognised by $e_{h, k}$ have length $k$.
        The $i$-th block
        $a_m^{k-i-1} \cdot (a_1 + \cdots + a_{h_i}) \cdot \Sigma_m^i$
        contains $h_i \cdot m^i$ words
        since there are $h_i$ distinct options to choose a letter from
        $a_1 + \cdots + a_{h_i}$
        and $m^i$ distinct options to choose a word from $\Sigma_m^i$.
        The languages of different blocks are disjoint,
        since the prefix $a_m^{k-i-1}$ uniquely determines the block
        (since $h_i < m$ and thus $a_1 + \cdots + a_{h_i}$ does not contain $a_m$).
        Similarly, the two concatenations in a block are unambiguous,
        and thus $e_{h, k}$ is unambiguous.
	\end{proof}
    The sought expression of measure $c$ is
	\begin{align}
	    e =
	        \Sigma_m^{\leq k - 1}
	        + e_{c_k,k}
	        + e_{d_1, k + 1} + \cdots
	        + e_{d_{j_1-1}, k + j_1 - 1}
	        + e_{\gamma, k + j_1 - 1 + l} \cdot e_{1, l}^*.
	\end{align}
	\begin{claim}
	    The expression $e$ is well defined provided $k \geq l \cdot \frac {\log (n+1)} {\log m} - 1 + l$.
	\end{claim}
	\begin{proof}[Proof (of the claim)]
	    The $e_{d_j,k+j}$'s are well defined since $d_j \leq n \leq m^k \leq m^{k+j}$
	    and thus there are $d_j$ words of length $k+j$ over alphabet $\Sigma_m$.)
	    In order for $e_{\gamma, k + j_1 - 1 + l}$ to be well defined we need
	    $\gamma \leq m^{k + j_1 - 1 + l}$.
	    Since $\gamma \leq (n+1)^l$ and $n \leq m^k$,
	    it suffice to require $(n+1)^l \leq m^{k - 1 + l}$.
	    The latter condition follows from the assumption.
	\end{proof}
	
	\begin{claim}
	    The expression $e$ is unambiguous and $\mu_\text{coin}(\lang e) = c$.
	\end{claim}
	\begin{proof}[Proof (of the claim)]
	    Unambiguity follows from the fact that
	    \begin{inparaenum}[1)]
	        \item all disjuncts of $e$ recognise disjoint languages,
    	    since $\Sigma_m^{\leq k - 1}$ recognises only words of length $\leq k-1$,
    	    $e_{c_k,k}$ only words of length $k$,
    	    $e_{d_1, k + 1}$ only words of length $k+1$, \dots,
    	    $e_{d_{j_1-1}, k + j_1 - 1}$ only words of length $k + j_1 - 1$, and
    	    $e_{\gamma, k + j_1 - 1 + l} \cdot e_{1, l}^*$ only words of length $ \geq k + j_1 - 1 + l$, and
	        \item all disjuncts $e_{c_k,k}, e_{d_1, k + 1}, \dots, e_{d_{j_1-1}, k + j_1 - 1}$
	        are unambiguous by \cref{claim:ehk},
	        and $e_{\gamma, k + j_1 - 1 + l} \cdot e_{1, l}^*$ is unambiguous
	        since $e_{1, l}^*$ contains only one word of length $l$.
	    \end{inparaenum}
	    In general $\mu_\text{coin}(L^k) = (n+1)^{k-1} \mu_\text{coin}(L)^k$
	    (by \eqref{eq:mu:LM}, when such power is unambiguous)
	    and thus $\mu_\text{coin}(L^*) = \mu_\text{coin}(L^0 \cup L^1 \cup \cdots) = \mu_\text{coin}(L^0) + \mu_\text{coin}(L^1) + \cdots = \frac 1 {n+1} (((n+1)\mu_\text{coin}(L))^0 + ((n+1)\mu_\text{coin}(L))^1 + \cdots) = \frac 1 {n+1} \cdot \frac 1 {1 - (n+1)\mu_\text{coin}(L)}$ when the Kleene iteration is unambiguous.
	    Since $\mu_\text{coin}(\lang {e_{1, l}}) = \frac 1 {(n+1)^{l+1}}$,
	    we obtain $\mu_\text{coin}(\lang {e_{1, l}^*}) = \frac 1 {n+1} \cdot \frac 1 {1 - (n+1)\frac 1 {(n+1)^{l+1}}} = \frac {(n+1)^{l-1}} {(n+1)^l - 1}.$
	    Consequently,
	    \begin{align*}
	        \mu_\text{coin}(\lang {e_{\gamma, k + j_1 - 1 + l} \cdot e_{1, l}^*}) 
	            &= (n+1) \cdot \mu_\text{coin}(\lang {e_{\gamma, k + j_1 - 1 + l}}) \cdot \mu_\text{coin}(\lang{e_{1, l}^*}) = \\
	            &= (n+1) \cdot \frac \gamma {(n+1)^{k + j_1 + l}} \cdot \frac {(n+1)^{l-1}} {(n+1)^l - 1} = \\
	            &= \frac {\gamma} {(n+1)^{k+j_1}} \cdot \frac 1 {(n+1)^l - 1}.
        \end{align*}
        By \eqref{eq:beta:next}, $\mu_\text{coin}(\lang e) = \mu_\text{coin}(\Sigma_m^{\leq k - 1}) + \mu_\text{coin}(\lang{e_{c_k, k}}) + \mu_\text{coin}(\lang{e_{d_1, k + 1}}) + \cdots + \mu_\text{coin}(\lang{e_{d_{j_1-1}, k + j_1 - 1}}) + \mu_\text{coin}(\lang {e_{\gamma, k + j_1 - 1 + l} \cdot e_{1, l}^*}) = c$,
        as required.
	\end{proof}
    \begin{claim}
        \label{claim:size}
        The expression $e$ has size $O((n \log q +j_1+l+1)^3)$.
    \end{claim}
    \begin{proof}[Proof (of the claim)]
        The size of $e_{h, k}$ is $\card {e_{h, k}} = O(k \cdot (m+k)) = O(k \cdot (n+k))$,
        since $m \leq n$.
        Thus the size of $e$ is $\card e \leq O((k+n+j_1+l+1)^3)$.
        By \eqref{eq:k}, $\card e \leq O((n \log q +j_1+l+1)^3)$, as required.
    \end{proof}

    If we further assume that $q \divides (n+1)^\ell$ for some $\ell \in \N$,
    then we argue that we can take $e$ to have size polynomial in $\log q$, $n$, and $\ell$.
	By \eqref{eq:k} and the additional assumption on $q$,
	\begin{align}
		\label{eq:k:small}
		k = O(n \ell \log (n+1)).
	\end{align}
	We can write $d = \frac r s$ with $r, s \in \N$ and $r \leq s$ relatively prime,
	where $s \divides (n+1)^{k+1}$.
	If we decompose the base $n+1$ in prime factors as
	$n+1 = p_1^{z_1} p_2^{z_2} \cdots p_m^{z_m}$
	with $z_1, \dots, z_m \in \N$,
	then $s$ is of the form $s = p_1^{t_1} p_2^{t_2} \cdots p_m^{t_m}$ with $t_i \leq (k+1)z_i$.
	By \cite[Theorem 136]{HardyWright:Numbers:2008}%
	\footnote{Strictly speaking, \cite[Theorem 136]{HardyWright:Numbers:2008} requires that the base $n+1$ be a product of primes. However, as explained in the paragraph preceding the statement of the theorem,
	no difficulty arises if $n+1$ has square divisors; c.f.~also the first footnote on page 145.},
	$d$ can be written in base $(n+1)$ with a \emph{finite} expansion of length
	$j_1 = \max\set{\frac {t_1} {z_1}, \dots, \frac {t_m} {z_m}} = O(k+1)$.
	In particular, the period has zero length $l = 0$.
	By \eqref{eq:k:small}, $j_1 = O(n \ell \log (n+1))$.
	By applying \cref{claim:size} to this case, we obtain a regular expression $e$ of size
	%
		$O((n \log q +j_1+l+1)^3) \leq O((n \log q +n \ell \log (n+1) +1)^3)$
	%
	which is polynomial in $\log q = O(\ell \log n)$, $\ell$, and $n$, 
	as required.
\end{proof}
\end{document}